\documentclass[conference,a4paper]{IEEEtran}
\pdfoutput=1
\usepackage[pdftex]{graphicx} 
\usepackage{calc} 
\usepackage{enumitem} 

\usepackage[all]{nowidow} 
\usepackage[protrusion=true,expansion=true]{microtype} 

\usepackage{amsmath}
\usepackage{amssymb,amsthm}
\usepackage{amsfonts}
\usepackage{cite}
\usepackage{tikz}
\usepackage{overpic}
\usepackage{float}
\usepackage{bm}
\usepackage{bbm}
\usetikzlibrary{positioning,arrows,shapes,chains,fit,scopes}
\usepackage{booktabs}
\usepackage{colortbl}
\usepackage{color}
\usepackage{blkarray, bigstrut}
\usepackage{amsmath}

\theoremstyle{plain}
\newtheorem{theorem}{Theorem}
\newtheorem{lemma}{Lemma}

\theoremstyle{definition}
\newtheorem{definition}{Definition}
\newtheorem{example}{Example}
\newtheorem{construction}{Construction}
\theoremstyle{remark}

\newcommand{\mc}[1]{\mathcal{#1}}

\newcommand{\Mod}[1]{\ \mathrm{mod}\ #1}

\title{Reliable Throughput of Generalized Collision Channel without Synchronization}
\author{\IEEEauthorblockN{Yijun Fan, Yanxiao Liu, Yi Chen\IEEEauthorrefmark{2}, Shenghao Yang, and Raymond W. Yeung}

\thanks{Y.~Fan, Y.~Liu and R.~W.~Yeung are with the Department of Information Engineering, The Chinese
University of Hong Kong. 
Y.~Chen and S.~Yang are with the School of Science and Engineering, The Chinese University of Hong Kong, Shenzhen, Shenzhen, China, and also with the Shenzhen Research Institute of Big Data, Shenzhen,
Guangdong, 518172, China. 

The work of S.~Yang was supported in part by the National Key R\&D Program of China under Grant 2022YFA1005000 and by Shenzhen STIC under Grant JCYJ20180508162604311.

The work of Y.~Chen was supported in part by the Guangdong Research Project No. 2017ZT07X152 and by the Guangdong Provincial Key Laboratory of Future Networks of Intelligence with Grant No. 2022B1212010001. 

The work of R.~W.~Yeung was supported in part by a fellowship award from the RGC of Hong Kong SAR, China under Grant CUHK SRFS2223-4S03.

\IEEEauthorrefmark{2}Corresponding author (email: yichen@cuhk.edu.cn).
}
}

\IEEEoverridecommandlockouts
\begin{document}
\maketitle

\begin{abstract}
We consider a generalized collision channel model for general multi-user communication systems, an extension of Massey and Mathys' collision channel without feedback for multiple access communications. In our model, there are multiple transmitters and receivers sharing the same communication channel. 
The transmitters are not synchronized and arbitrary time offsets between transmitters and receivers are assumed. 
A ``collision" occurs if two or more packets from different transmitters partially or completely overlap at a receiver. Our model includes the original collision channel as a special case.

This paper focuses on reliable throughputs that are approachable for arbitrary time offsets.  
We consider both slot-synchronized and non-synchronized cases and characterize their reliable throughput regions for the generalized collision channel model. These two regions are proven to coincide. Moreover, it is shown that the protocol sequences constructed for multiple access communication remain ``throughput optimal" in the generalized collision channel model.
We also identify the protocol sequences that can approach the outer boundary of the reliable throughput region.


\end{abstract}

\section{Introduction}
Practical communication systems are often multi-user in nature. There are multiple transmitters and receivers sharing the same communication medium. Consequently, signals intended for a receiver may cause interference at other receivers. The capacity region of a general multi-user system has been a long-standing open problem in information theory. Yet there is a major need to design simple and efficient transmission strategies for these systems and understand their performance limits. This paper will not deal with multi-user information theory but address the above issues under an assumption on how the communication system operates, namely, under a so-called generalized collision channel model.

 The \emph{collision channel} model was originally proposed for multiple access communications, where a common receiver is accessed by multiple transmitters over a shared channel. It is the simplest and nontrivial multi-user system. In the canonical collision channel model \cite{ephremides1998information}, time is divided into timeslots of equal duration. The signals transmitted by the transmitters are modelled by fixed-length packets, each of which fits within one timeslot. A ``collision" occurs if two or more packets overlap at the receiver, and as a result, all collided packets cannot be decoded correctly. Under this model, the slotted ALOHA was developed, where each transmitter sends out a packet independent of the others and thus packet collision cannot be avoided. To have collision-free transmission, time-division multiple access (TDMA) can be used but requires stringent time synchronization among the transmitters and receiver which is sometimes impractical. 

To explore how much transmission capacity would be lost in the absence of synchronization, Massey and Mathys introduced the \emph{collision channel without feedback}~\cite{massey1985collision}. It is assumed that the transmitters cannot synchronize their transmissions as their relative time offsets are unknown to each other and can never be learned on their own due to lack of feedback. They showed that reliable multiple access communication is indeed possible regardless of the time offsets and they characterized the reliable capacity region (which we call the \emph{reliable throughput region} in this paper) of the collision channel without feedback. A key concept in their scheme is the \emph{protocol sequence}, which is a periodic binary sequence assigned to each transmitter indicating when it transmits (sequence values equal 1) and when it keeps silent (sequence values equal 0). The protocol sequences they designed have the ``shift-invariant" property that the throughput of each link (the amount of successful packet transmissions per timeslot) is constant at all possible time offsets, and hence ensure reliable communication. 

Following~\cite{massey1985collision}, many other families of protocol sequences for multiple access communication were studied, such as the shift-invariant sequences \cite{shum2009shift}, the Wobbling sequences~\cite{wong2007new}, and the Chinese Remainder Theorem sequences~\cite{chen2018crt}. 
In addition, some extensions of Massey and Mathys’ collision channel without feedback have also been considered  \cite{tinguely2005collision,bae2014outerbdmr,zhang2016protocol,zhang2020throughtputone}. 
All of the aforementioned works focus on multiple access communications.
In this paper, we consider a general multi-user system and introduce a generalized collision channel model without feedback. In our model, multiple transmitters are paired with multiple receivers, and each pair causes interference to some other pairs that are specified by a collision profile. Arbitrary time offsets are assumed between transmitters and receivers, and are unknown in advance to all the transmitters and receivers. 
Different transmitters and receivers may correspond to one physical node, which provides flexibility in the network structure. When all receivers actually correspond to the same physical node, the generalized collision channel is reduced to the original model in~\cite{massey1985collision}.

We consider both slot-synchronized and non-synchronized communications without feedback, and prove that their reliable throughput regions for the generalized collision channel model are the same.
We characterize the reliable throughput region regardless of the time offsets, which adds to our understanding of the performance limits of multi-user systems. It is also shown that the protocol sequences constructed for multiple access communication in \cite{massey1985collision} can still approach every point in the reliable throughput region of the generalized model.
In particular, the protocol sequences that can approach the outer boundary of the reliable throughput region are identified.

The organization of this paper is as follows. Section~\ref{sec::model} describes the network model including the generalized collision channel without feedback.
Section~\ref{sec::region} analyzes the reliable throughput regions of slot-synchronized and non-synchronized cases. 
Section~\ref{sec::outerbd} investigates the points on the outer boundary of reliable throughput region. Finally, we give a concluding remark in Section \ref{section:conclusion} and further extend our model in Section~\ref{sec::model}.

Throughout this paper, we use $\mathbb{Z}, \mathbb{Z}_+$ and $\mathbb{Z}_{++}$ to denote integers, non-negative integers, and positive integers, respectively. We use $A(i,j)$ to denote the entry at the $i$-th row and $j$-th column of a matrix $\bm A$, and write $\bm A=[A(i,j)]$.
For two matrices $\bm A=[A(i,j)]$ and $\bm B=[B(i,j)]$, we write $\bm A\geq \bm B$ if $A(i,j) \geq B(i,j)$ for all $i$ and $j$, and $\bm A > \bm B$ if $A(i,j) > B(i,j)$ for all $i$ and $j$. We write $\text{diag}()$ for the diagonal matrix whose diagonal entries are the arguments inside the parentheses. 

\section{Network Model}
\label{sec::model}
\subsection{Collision Channel Model without Feedback}
Consider a general communication system 
with $M$ transmitters, $u_1,\ldots, u_M$, paired up with $M$ receivers, $r_1,\ldots,r_M$. Every transmitter $u_i$ intends to send messages to receiver $r_i$, forming a communication link $l_i:=(u_i,r_i)$. The collection of all links forms the link set $\mc L:= \{l_1,l_2,\ldots,l_M\}$. In practice, some transmitters and receivers may correspond to an identical physical node. For example, when all receivers correspond to a single physical receiver, the system becomes a multiple access communication system. When transmitter $u_{i+1}$ and receiver $r_{i}$ are the same physical node for $i=1,\ldots, M-1$ respectively, the system constitutes a line network, as illustrated by Fig.\ref{fig::network_structure}.

\begin{figure}
    \centering
    {\includegraphics[width=0.9\linewidth]{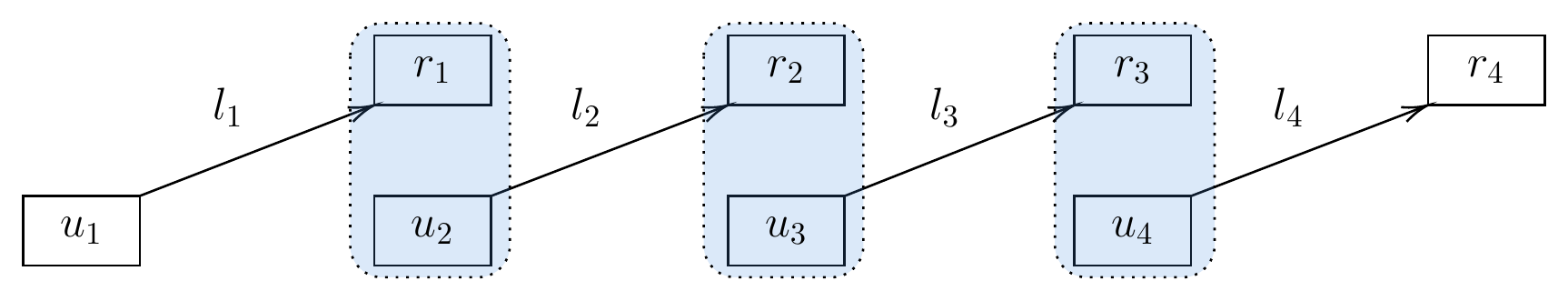}}    
    \caption{Line network. Receiver $r_i$ and transmitter $u_{i+1}$ correspond to one physical node, depicted as one shaded area. Solid arrows represent links.}
\label{fig::network_structure}
\end{figure}

All the $M$ links share the same radio resource and some of them can interfere with one another. When two interfering links are used to transmit packets simultaneously, a ``collision" occurs. We define a generalized collision channel model as follows. For each link $l_i$, let $\mc I(l_i)$ denote the set of links that may cause collisions against $l_i$. 
The collection of the collision set of all links forms the collision profile $\mc I:=\{\mc I(l):l\in\mc L\}$. For multiple access communication considered in \cite{massey1985collision}, we have $\mc I(l)=\mc L\setminus\{l\}$ for any link $l$. For a line network, with the assumption that the signal can propagate over a distance of two hops and is not detected further away, we have $\mc I(l_i)=\{l_{i-1},l_{i+1},l_{i+2},l_{i+3}\}$ for $i=2,\ldots, M-3$. 
To better represent the collision profile, we can use a directed graph, where the vertices are the links, and an edge from vertex A to vertex B exists if A is in B's collision set. 
Fig.~\ref{fig::collision_graph} shows the collision graph of an example line network. 
Without loss of generality, we assume that the collision graph is weakly connected. Otherwise, it should contain some disconnected subgraphs, each of which can be discussed separately. 
As a consequence, the collision graph has no isolated vertex, that is, every link $l$ has a nonempty collision set $\mc I(l)\neq \emptyset$ or falls in a collision set of some other link $l'\in\mc L$ with $l\in\mc I(l')$.

\begin{figure}
	\centering
 \begin{tikzpicture}[al/.style={draw, circle,inner sep=1pt,minimum size=0pt}]
  \node[al] (a) {$l_1$};
  \node[al,right=12mm of a] (b) {$l_2$};
  \node[al,right=12mm of b] (c) {$l_3$};
  \node[al,right=12mm of c] (d) {$l_4$};
  \draw[out=30,in=150,->] (a) to (b);
  \draw[out=30,in=150,->] (b) to (c);
  \draw[out=30,in=150,->] (c) to (d);
  \draw[out=210,in=-30,->] (d) to (c);
  \draw[out=210,in=-30,->] (c) to (b);
  \draw[out=210,in=-30,->] (b) to (a);
  \draw[out=120,in=60,->,distance=7mm] (d) to (a);
  \draw[out=-120,in=-60,->,distance=6mm] (d) to (b);
  \draw[out=-120,in=-60,->,distance=6mm] (c) to (a);
\end{tikzpicture}
 \caption{Collision graph of line network, in which signal can propagate over two hops at most. Each link $l_i$ corresponds to a transmitter-receiver pair $(u_i,r_i)$ in Fig.~\ref{fig::network_structure}.}
\label{fig::collision_graph}
\end{figure}
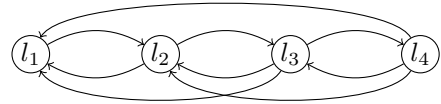

In our collision model, receiver $r_i$ can only detect signals from transmitter $u_j$ with $l_j\in\mc I(l_i)\cup\{l_i\}$. Define an index set $J(i):=\{j: l_j\in\mc I(l_i)\}\cup \{i\}$ for all $i=1,\ldots,M$.
We assume that there is an arbitrary time offset $\delta_i^j$ from $u_j$ to $r_i$ for every $i=1,\ldots, M$ and $j\in J(i)$. It means when transmitter $u_j$ sends a packet at time $t$, receiver $r_i$ will receive this packet at time $t+\delta_i^j$. Such time offsets can model the transmission delay of packets over distance. 
We assume without loss of generality that time is divided into timeslots of equal duration $1$ and a packet can be transmitted in an integer number of timeslots. 

It is possible that not all transmitter's clocks are synchronized. We will show that, however, it is sufficient to consider only the time offsets between transmitters and receivers. Suppose there is a common reference time (CRT) and the difference between $u_j$'s clock and CRT is $d_j$: when $u_j$ starts to transmit at time $t_j$, it is $t_j-d_j$ in CRT. Then a receiver $r_i$ starts to receive packets from $u_j$ at CRT $t_j-d_j+\delta_i^j$. Because of the arbitrariness of time offsets in our model, the observations of all the receivers $r_i$ with $j\in J(i)$ at CRT $t_j-d_j+\delta_i^j$ are not changed if the clock at $u_j$ is synchronized with CRT and the time offset $\delta_i^j$ is changed to $\delta_i^j-d_j$. Therefore, we can assume without loss of generality that all transmitter's clocks are synchronized with CRT. We further assume that the transmission of a packet always starts at the beginning of a timeslot.

We then consider two cases of the possible values of $\delta_i^j$: (1) slot-synchronized cases when all $\delta_i^j$ are integers, (2) non-synchronized cases when any $\delta_j^i$ is an arbitrary real number. 
For notation simplicity, we use a vector $\bm \delta = [\delta_i^j]_{i=1,\ldots,M\text{ and } j\in J(i)}$ to denote all the time offsets.


The collision channel models the situation where packets overlap at any receiver will result in a collision and the colliding packets cannot be decoded correctly. To be more precise, at any receiver $r_i$, a packet sent on link $l_i$ at time $t_i$ will collide with a packet sent on link $l_j\in\mc I(l_i)$ at time $t_j$ if and only if $$|(t_i+\delta_{i}^i)-(t_j+\delta_{i}^j)| < 1.$$ In the case of slot-synchronization, collisions occur only when packets completely overlap in some timeslots. While in the case of non-synchronization, partial overlap of packets in timeslots can also lead to collisions. The packets transmitting on links without collision should be received and decoded correctly at the corresponding receivers.

The time offsets are unknown to the transmitters in advance. Moreover, since there is no feedback, the time offsets cannot be learned during the communication. The purpose of this paper is to characterize the reliable communication throughput region of this generalized collision channel for arbitrary fixed time offsets.

\subsection{Protocol Signal and Reliable Throughput Region}
At any time $t$, a link can have only two possible statuses: active or inactive. Therefore for every link $l_i$, its status is completely determined by a binary protocol signal $s_i(t)$ assigned to transmitter $u_i$, that is, $l_i$ is active if $s_i(t)=1$ and inactive if $s_i(t)=0$. We define timeslot $n\in\mathbb Z$ as the semi-open time interval $n\leq t<n+1$.
Since time is slotted and all transmitters align their packet transmissions within timeslots, the links will always be active or inactive in a timeslot. 
Following \cite{massey1985collision} for analytical convenience, we consider that, for every protocol signal $s_i(t)$, there exist $0<L_i<\infty$, such that $s_i(t)=s_i(t+L_i)$ for all $t\in\mathbb{R}$. We call such an $L_i$ a period of $s_i(t)$.
Since time is divided into slots of duration $1$, $L_i$ is also an integer. Then the protocol signals of all links share a common period $L=\prod_{i:l_i\in\mc L}L_i$. For all $i=1,\ldots,M$, we have $s_i(t) = s_i(t+L)$.

In addition, we define the \textit{duty factor} $f_i$ for protocol signal $s_i(t)$ as the fraction of its nonzero period, i.e., $f_i$ measures the fraction of time during which link $l_i$ is used to transmit packets. Clearly, $0\leq f_i\leq 1$. We further define a duty factor vector $\bm f:=[f_1,\ldots,f_M]$.

Next, we define the \textit{reliable throughput region} of the generalized collision channel. It means the set of all throughput vectors $\bm T=[T_1,T_2,\ldots,T_M]$, with $T_i\geq 0$ packets per timeslot for link $l_i$, that are approachable for any fixed time offsets. By \textit{approachable}, we mean that given any positive number $\epsilon$, there exist protocol signals $s_i(t)$ for each transmitter $u_i$ such that the receiver $r_i$ is able to decode the packets from $u_i$ at a rate no smaller than $T_i-\epsilon$ packets/timeslot, for any fixed values of the time offsets $\bm \delta$. We use $\mc C_s$ and $\mc C_u$ to denote the reliable throughput region in slot-synchronized cases (times offsets are integers) and non-synchronized cases (time offsets are real numbers), respectively. 

The main result of this paper is that we explicitly characterize $\mc C_u$ and $\mc C_s$ and show that $\mc C_u=\mc C_s$. Moreover, given any $\bm T\in \mc C_u$ or $\bm T\in \mc C_s$, we provide a way to construct the protocol signal that can approach it.

It needs to be mentioned that if the time offsets $\bm \delta$ are known in advance, it is possible to design joint protocol signals that are collision-free \cite{ma2021rate, fan2021continuity}. In such collision-free cases, $\bm T=\bm f$ and the reliable throughput region is convex. However, when $\bm \delta$ changes, the customized protocol signals are no longer guaranteed to be collision-free, leading to $\bm T\leq \bm f$. Our definition of throughput can be viewed as the worst performance of protocol signals at any time offsets. If the time offsets are unknown to all links and cannot be learned from feedback, the reliable throughput region of a multiple access channel and its outer boundary has been characterized in several previous works \cite{massey1985collision,zhang2016protocol,zhang2020throughtputone,shum2009shift}. In the following discussion, we extend the analysis to the generalized collision channel model. 

\section{Reliable throughput region}
\label{sec::region}
In this section, we characterize the reliable throughput regions $\mc C_s$ and $\mc C_u$. We first define a region $\mc C$ and then prove that $\mc C = \mc C_s = \mc C_u$.
\begin{definition}\label{def:region}
Given a link set $\mc L=\{l_1,l_2,\ldots,l_M\}$ and its collision profile $\mc I$, $\mc C$ consists of the set of all points $\bm C=[C_1,\ldots,C_M]$, such that \begin{equation}
\label{eq::capacity_pt}
    C_i=f_i\prod_{j:l_j\in\mc I(l_i)}(1-f_j),
\end{equation} where $[f_1,f_2,\ldots,f_M]\in[0,1]^M$.
\end{definition}
In the definition above, $\bm f=[f_1,f_2,\ldots,f_M]$ can be regarded as the duty factor vector of a certain set of protocol signals. 
When $\mc I(l_i)=\mc L\setminus\{l_i\}$ for all $i=1,2,\ldots,M$, which corresponds to the original collision channel model in \cite{massey1985collision}, (\ref{eq::capacity_pt}) is $C_i= f_i\prod_{j=1,j\neq i}^M (1-f_j)$, reducing to the formula in \cite[Theorem~1]{massey1985collision}. 

We can see that $\mc C$ is downward comprehensive. That is, if $\bm C\in \mc C$, any $\bm C'$ satisfying $\bm 0\leq \bm C'\leq \bm C$ also belongs to $\mc C$, where $\bm 0$ is the all-zero vector of compatible size. 
To prove $\mc C=\mc C_u=\mc C_s$, we show in Lemma~\ref{thm::regionequality} that the throughput vectors outside $\mc C$ are not approachable. 
\begin{lemma}
\label{thm::regionequality}
    Under the same $\mc L$ and $\mc I$, $\mc C_u\subseteq \mc C_s\subseteq\mc C$.
\end{lemma}
\begin{proof}
By the definitions of $\mc C_s$ and $\mc C_u$, and the fact that the possible values of time offsets for the slot-synchronized case are a subset of those for the non-synchronized case, we directly have $\mc C_u\subseteq \mc C_s$. In the rest of the proof, we only need to focus on the synchronized case and show that $\mc C_s\subseteq\mc C$.

For the purposes only of our proof, we impose a fictitious probability distribution on the time offsets $\delta_j^i$ for all $i=1,\ldots,M$ and all $j$ such that $j=i$ or $l_i\in\mc I(l_j)$. They are independent and identically distributed random variables that are equally likely to take on any values of $0,1,\ldots, L$. 

Given any choice of the protocol signals, $s_i(t)$ for $i=1,\ldots,M$, collectively denoted by $\bm s$, it follows from \eqref{equ:periodic} the periodicity of $s_i(t)$ and from the definition of duty factor $f_i$, that 
\begin{equation}\label{equ:duty}
\mathbf E[s_i(t-\delta_j^i)] = f_i
\end{equation}
for every time instant $t$. Therein the expectation $\mathbf E[\cdot]$ is taken with respect to $\delta_j^i$. 

Recall that $\bm\delta$ is used to collectively denote all $\delta_j^i$. Given a realization of the time offsets $\bm \delta$, every receiver $r_i$ at its time $t$ receives packets from its intended transmitter $u_i$ without collision if and only if $s_i(t-\delta_i^i)\prod_{j:l_j\in\mc I(l_i)}(1-s_j(t-\delta_i^j))=1$. Consider a time period $[t_0,t_0+L)$ for some $t_0$ being the boundary of a timeslot at the receiver. Recall that all transmitters align their packet transmissions within timeslots and thus collisions will happen only when received packets completely overlap. The effective throughput, denoted by $T^e_i(\bm s,\bm\delta)$, of link $l_i$ can be computed as follows
 \begin{align}\label{effectivethp}
    T^e_i(\bm s, \bm\delta):=\frac{1}{L}\int_{t=t_0}^{t_0+L} s_i(t-\delta_i^i)\prod_{j:l_j\in\mc I(l_i)} (1-s_j(t-\delta_i^j)) dt.
\end{align} 
Taking $\bm\delta$ into consideration, the expected throughput of link $l_i$ is\begin{align*}
    &\mathbf{E}[T^e_i(\bm s, \bm\delta)]\\
   =  &\mathbf{E}\left[\frac{1}{L}\int_{t=t_0}^{t_0+L}s_i(t-\delta_i^i)\prod_{j:l_j\in\mc I(l_i)}(1-s_j(t-\delta_i^j))dt\right]\\
    =& \frac{1}{L}\int_{t=t_0}^{t_0+L} \mathbf{E}\Big[s_i(t-\delta_i^i)\prod_{j:l_j\in\mc I(l_i)}(1-s_j(t-\delta_i^j))\Big]dt\\
    \overset{(a)}=& \frac{1}{L}\int_{t=t_0}^{t_0+L} \mathbf{E}\big[s_i(t-\delta_i^i)\big]\prod_{j:l_j\in\mc I(l_i)}(1-\mathbf{E}\big[s_j(t-\delta_i^j)\big])dt\\
    \overset{(b)}= &\frac{1}{L}\int_{t=t_0}^{t_0+L} f_i\prod_{j:l_j\in\mc I(l_i)}(1-f_j)dt\\
    = &f_i\prod_{j:l_j\in\mc I(l_i)}(1-f_j).
\end{align*}
Equality (a) holds as $\delta_j^i$ are independent and equality (b) holds from \eqref{equ:duty}. 

For any protocol signals $\bm s$, $\mathbf{E}[T^e_i(\bm s, \bm\delta)]$ tells the average throughputs of all links by taking into consideration all possible time offsets $\bm\delta$. Therefore there must exist a time offset $\bm \delta^*$, such that 
\begin{equation}\label{equ:upperbound}
T^e_i(\bm s, \bm\delta^*) \leq \mathbf{E}[T^e_i(\bm s, \bm\delta)]=f_i\prod_{j:l_j\in\mc I(l_i)}(1-f_j)
\end{equation}
for all links $l_i$. Let $\bm T^e(\bm s,\bm \delta)=(T^e_1(\bm s, \bm\delta),\ldots,T^e_M(\bm s, \bm\delta))$.

For any $\bm T=(T_1,\ldots,T_M)\in \mc C_s$, it must be approachable by some protocol signals, say $\bm s^*$. Since $\bm T$ is approachable regardless of the time offsets, \eqref{equ:upperbound} implies that $\bm T\leq \mathbf{E}[\bm T^e(\bm s^*, \bm\delta)]$ where the inequality is defined component-wise. Meanwhile we see that $\mathbf{E}[\bm T^e(\bm s^*, \bm\delta)]\in \mc C$, 
indicating $\bm T\in \mc C$. We have proved that $\mc C_s\subseteq \mc C$.
\end{proof}

In the rest of this section, we construct protocol sequences to approach all points in $\mc C$ for both slot-synchronized and non-synchronized cases. We first prove that $\mc C\subseteq \mc C_s$ through a constructive proof: given any $\bm C\in\mc C$, we can construct a set of protocol sequences that approaches $\bm C$. Then we prove $\mc C\subseteq \mc C_u$ by extending this construction. For the time being, we consider the special case that $\bm \delta$ takes integer values. 

Recall that the timeslot $n\in\mathbb Z$ is the time interval $n\leq t<n+1$ and $s_i(t)$ is either $0$ or $1$ within timeslot $n$. As a result, the protocol signal $s_i(t)$ can be equivalently represented by a binary periodic \emph{protocol sequence} $s_i:=[s_i(0),s_i(1),\ldots,s_i(L-1)]$ ($L\in\mathbb Z_{++}$ is the period). Link $l_i$ is used to transmit a packet in timeslot $n$ if $s_i(n\Mod{L})=1$, and is idle otherwise. On the receiver side, it will receive packets according to a shifted version of the protocol sequence, and the amount of the shift is determined by the time offset. For example, receiver $r_i$ will observe the packets from transmitter $u_j$ for $j\in J(i)$, according to the sequence $[s_j(-\delta_i^j),\ldots,s_j(L-\delta_i^j)]$, which is obtained from $s_j$ by $\delta_i^j$ right cyclic shifts.

Here we define two kinds of submatrices that will appear frequently in the rest of this section. For any $M\times L$ matrix $\bm A=[A(j,t)]$ and $i=1,\ldots,M$, we define the submatrix $\bm A^i:=[A(j,t)]_{j\in J(i)}$ 
and the shifted submatrix under integer time offsets $\bm A^i[\bm \delta]:= [A(j, (t-\delta_i^j) \Mod{L})]_{j\in J(i)}$.  
Below is a concrete example.
\begin{example}
\label{example::time_offset}
    Let $\bm A=\begin{bmatrix}
        1&0&1&0\\
        1&1&0&0\\
        0&1&0&1
    \end{bmatrix}.$ For receiver $r_2$, let $J(2)=\{2,3\}$, $\delta_2^2=1$ and $\delta_2^3=3$. Then $[A(2,t)]=[1,1,0,0]$ and $[A(3,t)]=[0,1,0,1]$. Since $L=4$, we can compute\begin{align*}
        [A(2, &(t-\delta_2^2)\Mod{4})] \\
        &= [A(2,3),A(2,0),A(2,1),A(2,2)] = [0,1,1,0]\\
        [A(3, &(t-\delta_2^3)\Mod{4})]\\
        &=[A(3,1),A(3,2),A(3,3),A(3,0)] = [1,0,1,0].
    \end{align*}
    Hence, we have\begin{align*}
        \bm A^2 = \begin{bmatrix}
        1&1&0&0\\
        0&1&0&1
        \end{bmatrix} \text{ and } 
         \bm A^2[\bm \delta] = \begin{bmatrix}
        0&1&1&0\\
        1&0&1&0
        \end{bmatrix}.
    \end{align*}
    Fig.~\ref{fig::time_offset} is a graphical explanation of $\bm A^i[\bm \delta]$.
\end{example}
\begin{figure}
    \centering
    {\includegraphics[width=.9\linewidth]{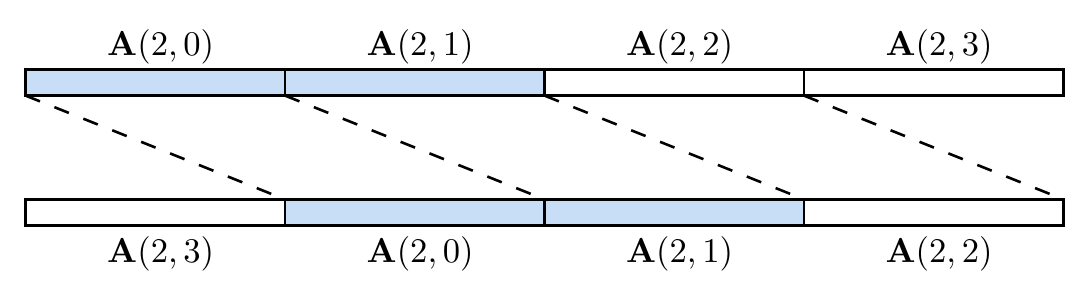}}
    \caption{The row $[A(j,(t+\delta_i^j)\Mod{L})]$ is obtained from $[A(j,t)]$ by $\delta_i^j$ right cyclic shifts. Each shaded slot represents an entry with a value $1$.}
\label{fig::time_offset}
\end{figure}

A \emph{protocol matrix} is defined by $\bm S:=[s_1, s_2,\ldots,s_M]^{\intercal}\in\{0,1\}^{M\times L}$, where the $i$-th row of $\bm S$ is the binary periodic protocol sequence $s_i$ for link $l_i$ with period $L$. In the slot-synchronized cases, when all transmitters follow the protocol matrix $\bm S$, the sequences observed by receiver $r_i$ will form the matrix $\bm S^i[\bm \delta]$. 
We follow the sequence construction scheme in \cite{massey1985collision} described in Construction~\ref{construction_1} and use $\bm S_{M,q}$ to denote the protocol matrix constructed by this scheme.

\begin{construction}\cite{massey1985collision}:
\label{construction_1}
Given rational duty factors $f_1,\ldots, f_M$, where $f_i=q_i/q$ for some $q_i\in\mathbb{Z}_{+}$ and $q\in\mathbb{Z}_{++}$, we use $\bm A_{M,q}$ to denote the $M\times q^M$ matrix whose $j$-th column is the $M$ place radix-$q$ representation of the integer $q^M-j$, with the least significant digit at the top. A protocol matrix $\bm S_{M,q}$ can be constructed by mapping, within the $i$-th row of $\bm A_{M,q}$, the $q$-ary digits $q-1,1-2,\ldots,q-q_i$ to value $1$ and mapping the $q$-ary digits $q-q_i-1,\ldots, 0$ to value $0$. 
\end{construction}

Here, $\bm S^i_{M,q}[\bm \delta]$, mapped from $\bm A^i_{M,q}[\bm \delta]$, denotes the row-cyclic-shifted submatrix of $\bm S_{M,q}$ observed at receiver $r_i$.
Reference \cite{massey1985collision} only considered multiple access communications where $J(i)=\{1,\ldots, M\}$ and $\delta_1^i=\ldots=\delta_M^i$ for all $i=1,\ldots, M$. As a result, we let $\bm A_{M,q}^1[\bm \delta]=\ldots=\bm A_{M,q}^M[\bm \delta]=\bm A_{M,q}[\bm\delta]$ and $\bm S_{M,q}^1[\bm \delta]=\ldots=\bm S_{M,q}^M[\bm \delta]=\bm S_{M,q}[\bm\delta]$. Massey and Mathys proved that for any $\bm \delta$, the columns of the row-cyclic-shifted matrix $\bm A_{M,q}[\bm \delta]$ are a permutation of those of $\bm A_{M,q}$, as are $\bm S_{M,q}[\bm \delta]$ and $\bm S_{M,q}$. We call this property of the sequences ``shift invariant". Using this property, \cite{massey1985collision} also showed that the reliable throughput satisfies $f_i\prod_{j:j\neq i}(1-f_j)$ for all $i$ and arbitrary integer time offsets $\bm \delta$.

However, in our generalized collision channel model, when $\bm S_{M,q}$ is used to instruct the transmission, the observed matrix $\bm S_{M,q}^i[\bm \delta]$ at receiver $r_i$ varies with $i$. Therefore the results of \cite{massey1985collision} cannot be applied directly. 
In the following lemma, we first prove a similar property that $\bm A^i_{M,q}[\bm \delta]$ ($\bm S^i_{M,q}[\bm \delta]$) is a column permutation of $\bm A^i_{M,q}$ ($\bm S^i_{M,q}$), and then show that the throughputs of both $\bm S^i_{M,q}$ and $\bm S^i_{M,q}[\bm \delta]$ are $f_i\prod_{j:l_j\in\mc I(l_i)}(1-f_j)$ for arbitrary time offsets $\bm \delta$, leading to the conclusion that any $\bm C\in\mc C$ is approachable and thus belong to $\mc C_s$.



\begin{lemma}
\label{lemma::perm}
    For any integer time offsets $\bm \delta$ and collision set $\mc I(l_i)$ at receiver $r_i$, $i=1,2,\ldots,M$, the columns of $\bm A^i_{M,q}[\bm \delta]$ consists of all $q^{|J(i)|}$ possible tuples in $\{0,1,\ldots,q-1\}^{|J(i)|}$, and each possible tuple appears exactly $q^{M-|J(i)|}$ times. Therefore, $\bm A^i_{M,q}[\bm \delta]$ is a column permutation of $\bm A^i_{M,q}$.
\end{lemma}
\begin{proof}
    

    Fix a receiver $r_i$, $\bm \delta$ and collision set $\mc I(l_i)$. For the purpose of the proof, we set time offsets $\delta_i^j=0$ for $j\notin J(i)$ that has no real meaning. Define $\bm \delta_i:=[\delta_i^j]_{j:l_j\in\mc L}$ for $i=1,\ldots,M$. Then for each $i$, $\bm A_{M,q}[\bm \delta_i]=[A_{M,q}(j, (t-\delta_i^j)\Mod{L})]$ is a row-cyclic-shifted matrix with size $M\times q^M$. Observe that $\bm A^i_{M,q}[\bm \delta]$ can be obtained by deleting the $j$-th row in $\bm A_{M,q}[\bm \delta_i]$ for $j\notin J(i)$. Our proof is by induction on the number of deleted rows. Note that the columns of $\bm A_{M,q}[\bm \delta_i]$ contain all possible $q$-ary $M$-tuples and every column in $\bm A_{M,q}[\bm \delta_i]$ appears once. 

    Let $\bm A_{M,q}^{[-m]}[\bm \delta_i]$ denote the matrix after deleting $m$ arbitrary rows from $\bm A_{M,q}[\bm \delta_i]$, for $0<m < M$.
    First, consider $m=1$. With $M-1$ rows, $\bm A_{M,q}^{[-1]}[\bm \delta_i]$ has at most $q^{M-1}$ distinct columns. Suppose there exists a column in $\bm A_{M,q}^{[-1]}[\bm \delta_i]$ that appears more than $q$ times. Then in $\bm A_{M,q}[\bm \delta_i]$, there must exist one column that appears more than once, leading to a contradiction. Hence, each column in $\bm A_{M,q}^{[-1]}[\bm \delta_i]$ can appear less than or equal to $q$ times. Furthermore, since there are in total $q^M$ columns of $\bm A_{M,q}^{[-1]}[\bm \delta_i]$, it must consists $q^{M-1}$ distinct columns, each appearing $q$ times.

    Assume the induction hypothesis is true for $\bm A_{M,q}^{[-m]}[\bm \delta_i]$. That is, $\bm A_{M,q}^{[-m]}[\bm \delta_i]$ consists $q^{M-m}$ distinct columns, each appearing $q^m$ times. We can use an argument similar to the previous one to show that this is also true for $\bm A_{M,q}^{[-(m+1)]}[\bm \delta_i]$, which is obtained by deleting one arbitrary row from $\bm A_{M,q}^{[-m]}[\bm \delta_i]$. With $M-m-1$ rows, $\bm A_{M,q}^{[-(m+1)]}[\bm \delta_i]$ has at most $q^{M-(m+1)}$ distinct columns. Suppose there exists one column with more than $q^{m+1}$ occurrences in $\bm A_{M,q}^{[-(m+1)]}[\bm \delta_i]$, then $\bm A_{M,q}^{[-m]}[\bm \delta_i]$ contains a column with more than $q^{m}$ occurrences, contradicting to the induction hypothesis. Together with the fact that $\bm A_{M,q}^{[-(m+1)]}[\bm \delta_i]$ has $q^M$ columns, we claim that $\bm A_{M,q}^{[-(m+1)]}[\bm \delta_i]$ has $q^{M-(m+1)}$ distinct columns, each appearing $q^{m+1}$ times.

    Since $\bm A^i_{M,q}[\bm \delta]$ falls to some $\bm A_{M,q}^{[-m]}[\bm \delta_i]$ with $m=M-|J(i)|$, we complete the proof.

\end{proof}
We can prove $\mc C\subseteq \mc C_s$ based on Lemma~\ref{lemma::perm}. For non-synchronized cases, i.e., $\bm \delta$ may not take integer values, we adjust $\bm S_{M,q}$ from Construction~\ref{construction_1} by replacing $0$ with $0^k$ and $1$ with $1^{k-1}0$, where $k\in\mathbb{Z}_{++}$ is arbitrary, $0^k$ and $1^{k-1}$ denote a string of $k$ zeros and a string of $k-1$ ones. Because of the arbitrariness of $k$, all the points in $\mc C$ are approachable using the adjusted protocol matrices for non-synchronized collision channels. Hence, we have $\mc C\subseteq \mc C_u$ as well, leading to the following result.

\begin{theorem}
    Under the same $\mc L$ and $\mc I$, $\mc C_s=\mc C_u=\mc C$.
    \label{thm::region_eq}
\end{theorem}
\begin{proof}
1) Synchronized case ($\mc C_s=\mc C$):  Consider any $\bm C\in\mc C$ that can be expressed as \eqref{eq::capacity_pt} for some rational numbers $f_i\in [0,1]$, $i=1,2,\ldots, M$. These rational numbers can then be written as $f_i=q_i/q$ for some $q_i\in\mathbb{Z}_{+}$ and $q\in\mathbb{Z}_{++}$. We use $f_1,\ldots,f_M$ as the duty factors and employ Construction \ref{construction_1} to construct $\bm A_{M,q}$ and $\bm S_{M,q}$. User $u_i$ transmits packets following the $i$-th row $s_i$ of $\bm S_{M,q}$.

By Lemma \ref{lemma::perm}, we know that for any integer time offsets $\bm \delta$, the columns of $\bm A^i_{M,q}[\bm \delta]=[A_{M,q}(j, (t-\delta_i^j)\Mod{q^M})]_{j\in J(i),  }$ consists of all $q^{|J(i)|}$ possible tuples in $\{0,1,\ldots,q-1\}^{|J(i)|}$, and each possible tuple appears exactly $q^{M-|J(i)|}$ times. Therefore, in $\bm A^i_{M,q}[\bm \delta]$, there are $q^{M-|J(i)|}q_i\prod_{j:l_j\in \mc I(l_i)}(q-q_j)$ columns in which the entry in row 
$[A_{M,q}(i, (t-\delta_i^i)\Mod{q^M})]_{ }$ is a digit larger than or equal to $q-q_i$, while the entry in each other row 
$[A_{M,q}(j, (t-\delta_i^j)\Mod{q^M})]^{ }$ ($l_j\in\mc I(l_i)$) is an integer less than $q-q_j$. Recall that from $\bm A^i_{M,q}[\bm \delta]$ to $\bm S^i_{M,q}[\bm \delta]$ follows the q-ary to binary digit mapping. As a result, in $\bm S^i_{M,q}[\bm \delta]$, there are $q^{M-|J(i)|}q_i\prod_{j:l_j\in \mc I(l_i)}(q-q_j)$ columns in which the entry in row $[S_{M,q}(i, (t-\delta_i^i)\Mod{q^M})]_{ }$ is $1$, while the entry in each other row $[S_{M,q}(j, (t-\delta_i^j)\Mod{q^M})]^{ }$ ($l_j\in\mc I(l_i)$) is $0$. This indicates that receiver $r_i$ receives exactly $q^{M-|J(i)|}q_i\prod_{j:l_j\in \mc I(l_i)}(q-q_j)$ successful packets from $u_i$. Finally, for every $\bm \delta$ and $i=1,\ldots,M$, we have
\begin{align*}
        T_i^e(\bm S^i_{M,q}[\bm \delta])&= \frac{1}{q^M}q^{M-|J(i)|}q_i\prod_{j:l_j\in \mc I(l_i)}(q-q_j)\\
        &= \frac{q_i}{q} \prod_{j:l_j\in\mc I(l_i)}\frac{q-q_j}{q}\\
        &= f_i\prod_{j:l_j\in\mc I(l_i)}(1-f_j)\\
        &=C_i,
    \end{align*} which proves that $\bm C\in \mc C_s$. 

Now let $\bm C'$ be any vector in $\mc C$ that be expressed as \eqref{eq::capacity_pt} for some $\bm f'=[f_1',\ldots,f_M']$ (not necessary to be rational numbers). Due to the dense nature of rational numbers, each open neighbourhood of $\bm f'$ contains vectors with only rational components. Moreover, since the mapping (\ref{eq::capacity_pt}) from $[0,1]^M$ to $\mc C$ is continuous, every open neighbourhood of $\bm C'$ contains another $\bm C\in\mc C$ that corresponds to $\bm f$ with only rational components. We have shown that $\bm C\in \mc C_s$. If $\bm C'$ is an interior point of $\mc C$, we can always find $\bm C$ in $\bm C'$'s neighbourhood such that $\bm C'<\bm C$, implying $\bm C'\in \mc C_s$. If $\bm C'$ is on the boundary of $\mc C$, for any given positive $\epsilon$, we can also find $\bm C$ in $\bm C'$'s neighbourhood such that $\bm C'-\epsilon \bm 1<\bm C$, where $\bm 1$ is an all one vector of compatible size. This implies that $\bm C'$ is approachable and thus $\bm C'\in \mc C_s$.

In summary, $\mc C\subseteq \mc C_s$. Together with $\mc C_s\subseteq \mc C$ shown in Lemma~\ref{thm::regionequality}, we have $\mc C_s= \mc C$.

2) Non-synchronized case ($\mc C_u=\mc C$):
  Since the components of $\bm\delta$ are arbitrary real numbers, at any receiver, the timeslot edge of packets transmitted by different users are generally not aligned with each other, and thus results in partial overlapping packets. Moreover, partial overlap of packets leads to collisions.
    We follow the method in \cite{massey1985collision} to reduce the non-synchronized case to the slot-synchronized case. 
    
    
    Given any $\bm T\in \mc C_s$, suppose that $\bm S_{M,q}$ is the protocol matrix that approaches $\bm T$ in the synchronized case. 
    Let $k$ be an arbitrary positive integer. Denote $0^k$ and $1^{k-1}$ as a string of $k$ zeros and a string of $k-1$ ones, respectively. We construct a new protocol matrix $\bm S^{(k)}_{M,q}$ from $\bm S_{M,q}=[s_1,\ldots,s_M]^{\intercal}$
    by replacing its $0$ with $0^k$ and $1$ with $1^{k-1}0$. In the non-synchronized case, each user $u_i$ transmits packet according to the binary sequence in the $i$-th row of $\bm S^{(k)}_{M,q}=[s_1^{(k)},\ldots,s_M^{(k)}]^{\intercal}$. Consider its submatrix $\bm S^{(k),i}_{M,q}=[s_j^{(k)}]^{\intercal}_{j\in J(i)}$. A key observation of user $u_i$'s packets at its intended receiver $r_i$ is that, the collision effect of other users on these packet transmissions (due to partial overlap) is equivalent to that in the synchronized case if each occurrence of $1^{k-1}0$ in $s_j^{(k)}$ for $l_j\in\mc I(l_i)$ is replaced by $1^k$. Such equivalent synchronized sequence is denoted by $ \bar{\bm S}^{(k),i}_{M,q}$. If we take $k$-th decimation of the columns of $ \bar{\bm S}^{(k),i}_{M,q}$, we get $k$ matrices of size $|J(i)|\times q^M$, among which $k-1$ matrices are $\bm S^i_{M,q}=[s_j]^{\intercal}_{j\in J(i)}$, while the rest matrix has the entries in the row corresponding to user $u_i$ all zero. Then for each of $k-1$ matrices, user $u_i$ can approach throughput $T_i$ regardless of the time offsets. Therefore, user $u_i$'s overall throughput is $(k-1)T_i/k$. For sufficiently large $k$, $\bm T$ is approachable in the non-synchronized case. Hence we have $\mc C_s\subseteq \mc C_u$ and thus $\mc C_u=\mc C_s=\mc C$.
\end{proof}

\begin{example}
\label{example::3-link}
Consider a collision channel model with link set $\mc L=\{l_1,l_2,l_3\}$, and collision sets $\mc I(l_1)=\{l_2,l_3\}, \mc I(l_2)=\{l_1\}, \mc I(l_3)=\{l_1\}$. The throughput region $\mc C_s=\mc C_u=\{[T_1,T_2,T_3]\}$ is defined by\begin{equation}
\begin{aligned}
    &T_1=f_1(1-f_2)(1-f_3),\\
    &T_2=f_2(1-f_1), \quad
    T_3=f_3(1-f_1),
    \label{eq::cregion_1}
\end{aligned}
\end{equation} for all $\bm f=[f_1,f_2,f_3]\in[0,1]^3$.

Given a duty factor vector $\bm f=[1/2,1/2,1/2]$, we have $q=2$, $M=3$, $L=q^M=8$, and the  throughput vector $\bm T=[1/8,1/4,1/4]$. Next we give the protocol sequence that approaches $\bm T$. By Construction~\ref{construction_1}, we have \begin{align*}
    \bm A_{3,2} = \begin{bmatrix}
        1&0&1&0&1&0&1&0\\
        1&1&0&0&1&1&0&0\\
        1&1&1&1&0&0&0&0
        \end{bmatrix}.
\end{align*} Since $q=2$, $\bm A_{3,2}$ is binary, $\bm S_{3,2}=\bm A_{3,2}$. When $\bm\delta=\bm 0$, receiver $r_1, r_2$ and $r_3$ will observe $\bm S^1_{3,2} = \bm S_{3,2}$,
\begin{align*}
    \bm S^2_{3,2} &= \begin{bmatrix}
        1&0&1&0&1&0&1&0\\
        1&1&0&0&1&1&0&0
    \end{bmatrix},\\
    \bm S^3_{3,2} &= \begin{bmatrix}
        1&0&1&0&1&0&1&0\\
        1&1&1&1&0&0&0&0
    \end{bmatrix}.
\end{align*}
When $\delta_1^1=1,\delta_1^2=2,\delta_1^3=3$, receiver $r_1$ will observe
\begin{align*}
   \bm S_{3,2}^1[\bm\delta] = \begin{bmatrix}
        0&1&0&1&0&1&0&1\\
        0&0&1&1&0&0&1&1\\
        0&0&0&1&1&1&1&0
        \end{bmatrix}.
\end{align*}
It can be checked that for any integer $\bm\delta$, the throughput vector is always $[1/8, 1/4, 1/4]$.
\end{example}

\section{Outer boundary of reliable throughput region}
\label{sec::outerbd}
We define the \textit{outer boundary} of a reliable throughput region $\mc C$ as the set of all points $\bm T\in \mc C$ such that there does not exist $\bm T'\in \mc C$ with $\bm T<\bm T'$. By Definition \ref{def:region}, every point in $\mc C$ must be mapped from a duty factor vector $\bm f$. 
In \cite{abramson1973packet, massey1985collision} for multiple access communications, it has been proved that every outer boundary point $\bm T$ is mapped from an $\bm f$ with $\sum_i^M f_i= 1$. The intuition is that the total fraction of time used by interfering users for packet transmission cannot exceed $1$. 

It is a natural guess that this correspondence also exists in the generalized collision channel model. That is, the points on the outer boundary of $\mc C$ are mapped from those $\bm f$ with $\min_{i=1,\ldots,M}\{\sum_{j\in J(i)}f_j\} = 1$. Considering Example \ref{example::3-link}, the conjecture is that all the points on the outer boundary of $\mc C$ are determined by an $\bm f$ satisfying $\min\{f_1+f_2,f_1+f_3, f_1+f_2+f_3\}=1$. However, we will see in the following example that this conjecture is false. 
\begin{example}[Counter-example to the outer boundary conjecture]
\label{example::counterex}
Continue on the network model defined in Example~\ref{example::3-link},
given that $T_2=T_3=1/4$ in $\bm T$, if the conjecture above was correct, the maximum $T_1$ should be $1/8$, mapped from $\bm f=[1/2,1/2,1/2]$.
We can, however, approach the throughput vector $\bm T' = [27/200, 1/4,1/4] \geq \bm T$ with $\bm f'=[3/8,2/5,2/5]$ by \eqref{eq::cregion_1}, implying that $\bm T$ is not on the outer boundary of $\mc C$.
\end{example}

In the rest of this section, we will derive a necessary and sufficient condition for the duty factors that map to points on the outer boundary of $\mc C$. 
Pick any link, say $l_1$, with nonempty collision set $\mc I(l_1)\neq\emptyset$. Observe that when the throughputs of other links are not smaller than some fixed values $T_2,\ldots, T_M$, such desirable duty factor vectors should maximize the throughput $T_1$. This observation can be formalized into the following optimization problem:
  \begin{align*} 
    \max_{\substack{[f_1,f_2,\ldots, f_M]\\ \in (0,1)^M}} \quad &f_1\prod_{i:l_i\in \mc I(l_1)}(1-f_i) \tag{OP1}\\
         s.t. \quad &f_k\prod_{i:l_i\in\mc I(l_k)}(1-f_i) \geq T_k, \quad k=2,\ldots,M.
    \end{align*}

The optimal solution set to (OP1) for all possible $[T_2, T_3,\ldots, T_M]$ is the set of duty factors that characterize the outer boundary. Hence,  we now focus on the optimality conditions of (OP1). 
We restrict $\bm f\in(0,1)^M$ to exclude the degenerated cases that the throughputs of some links are $0$.

Note that (OP1) is equivalent to the following problem
    \begin{align*}
   & \min_{\substack{[f_1,f_2,\ldots, f_M]\\ \in (0,1)^M}} \ -\ln f_1 - \sum_{i:l_i\in\mc I(l_1)}\ln(1-f_i) \tag{OP2}\\
  &       s.t. \   -\ln f_k- \sum_{i:l_i\in\mc I(l_k)}\ln(1-f_i) +\ln T_k\leq 0, k=2,\ldots,M.
    \end{align*}
%
\begin{lemma}
\label{lemma::kkt}
    (OP2) is a convex problem. If it is feasible, it satisfies Slater’s conditions. Hence the KKT conditions of (OP2) are sufficient and necessary for its optimal solution.
\end{lemma}
\begin{proof}
For notation simplicity, for $k=1,\ldots, M$, define functions 
$$g_k(\bm f)=-\ln f_k - \sum_{i:l_i\in\mc I(l_k)}\ln(1-f_i) +\ln T_k.$$ 
The objective function becomes $g_1(\bm f)$ and the constraint functions are $g_k(\bm f)$ for $k=2,\ldots, M$. They are all convex because their Hessian matrices are positive semi-definite, i.e., \begin{align*}
    \text{diag}\left(\frac{1}{f_k^2}, \left(\frac{1}{(1-f_i)^2}\right)_{i:l_i\in\mc I(l_k)}\right) \succeq \bm 0,
\end{align*} 
%
Hence, (OP2) is a convex optimization problem. 

If (OP2) is feasible, then there exist $0<f_1,\ldots,f_M<1$ such that $g_k(\bm f)+\ln T_k \leq 0$ for $k=2,\ldots, M$. Let us reduce $f_1$ by $\epsilon>0$ so that $f_1'=f_1-\epsilon>0$. With $\bm f'=[f_1',f_2,\ldots,f_M]$, we have $g_k(\bm f') +\ln T_k<g_k(\bm f) +\ln T_k \leq 0$ for $k:l_1\in\mc I(l_k)$, and $g_k(\bm f') +\ln T_k= g_k(\bm f) +\ln T_k\leq 0$ for $k:l_1\notin\mc I(l_k)$. So $\bm f'$ is feasible. We further reduce $f_k$ for $k:l_1\in\mc I(l_k)$ by $\epsilon'$. Let $\bm f''$ be such that $f''_k=f_k-\epsilon'$ for $k:l_1\in\mc I(l_k)$, $f''_k=f_k$ for $k:l_1\notin\mc I(l_k)$ and $f''_1=f_1'$. There exists $\epsilon'>0$ such that $\bm f''>\bm 0$ and $g_k(\bm f') +\ln T_k<g_k(\bm f'') +\ln T_k < 0$ for $k:l_1\in\mc I(l_k)$. Meanwhile, for links that have these renewed links in their collision set, we have $g_k(\bm f'') +\ln T_k < 0$. Since there is no isolated link, we can keep reducing the duty factors to reach a point $\bm f'''$ such that $g_k(\bm f''') +\ln T_k < 0$ for all $k=2,\ldots, M$. We conclude that Slater's condition is satisfied.

%
\end{proof}

The KKT points to (OP2), therefore, are those duty factor vectors corresponding to the outer boundary of the reliable throughput region. Let $\lambda_1=T_1=1$ and $\bm \lambda = [\lambda_1, \lambda_2,\ldots, \lambda_M]$, where $\lambda_i, i=2,\ldots, M$, are the Lagrangian multipliers. The KKT conditions are: for $k=2,\ldots,M$,
\begin{enumerate}
    \item $-\ln f_k-\sum_{i:l_i\in\mc I(l_k)}\ln (1-f_i)+\ln T_k\leq 0$,
    \item $\bm f\in (0,1)^M$,
    \item $\bm \lambda \geq \bm 0$,
    \item $\lambda_k(-\ln f_k-\sum_{i:l_i\in\mc I(l_k)}\ln (1-f_i)+\ln T_k)=0$, 
      \item $\nabla_{\bm f} L(\bm f, \bm \lambda)=\bm 0$, where $L(\bm f,\bm \lambda) = \sum_{k=1}^M \lambda_k (-\ln f_k - \sum_{i:l_i\in\mc I(l_k)}\ln(1-f_i)+\ln T_k)$.
\end{enumerate} 

Let $\bm F = \text{diag}(\bm f)$, $\bm I$ be the $M\times M$ identity matrix, and $\bm E=[e_{ij}]\in\{0,1\}^{M\times M}$ be the adjacency matrix of the collision graph, i.e., $e_{ij}=1$ if $l_i\in\mc I(l_j)$ and $e_{ij}=0$ otherwise. The following theorem demonstrates that whether a duty factor vector $\bm f$ corresponds to an outer boundary point of $\mc C$ or not depends on $\bm F$ and the adjacency matrix $\bm E$.

\begin{theorem}
\label{thm::df_outerbd}
   For any $\bm f \in(0,1)^M$, it determines a point on the outer boundary of reliable throughput region, if only if the Perron–Frobenius eigenvalue of $\bm F(\bm E+\bm I)$ is one.
    
\end{theorem}
\begin{proof}
It follows from $\partial L(\bm f, \bm \lambda)/\partial f_k = 0$ and $f_k>0$ that
\begin{equation}
\label{eq::outercond}
    \lambda_k = f_k \sum_{\substack{i:l_i\in\mc I(l_k)\cup\{l_k\}}} \lambda_i,\quad k=1,2,\ldots,M.
\end{equation} 
Write these $M$ equations into a compact form and we have \begin{equation}\label{equ:KKTder}
    \bm{\lambda} = \bm F(\bm E+\bm I)\bm{\lambda},
\end{equation}
which is equivalent to the KKT condition $\nabla_{\bm f}L(\bm f, \bm \lambda)= \bm 0$. 

Note that $\bm F(\bm E+\bm I)$ is a non-negative matrix. 
According to Perron-Frobenius Theorem\cite{axelsson_1994}, for a non-negative matrix, there exists an eigenvector with non-negative components and the corresponding eigenvalue (called Perron-Frobenius eigenvalue) is non-negative and greater than or equal, in absolute value, to all other eigenvalues. To satisfy \eqref{equ:KKTder}, $\bm F(\bm E+\bm I)$ must have the Perron-Frobenius eigenvalue equal to one.

On the other hand, if the Perron-Frobenius eigenvalue $\bm F(\bm E+\bm I)$ is one, then the (3) and (5) KKT conditions are satisfied. Meanwhile, we can find a set of $T_2,\ldots,T_M$ such that the (1) and (4) KKT conditions are also satisfied. For example, let $T_k = f_k\prod_{i:l_i\in \mc I(l_k)}(1-f_i)$ for all $\lambda_k>0$ and $T_k \leq f_k\prod_{i:l_i\in \mc I(l_k)}(1-f_i)$ for all $\lambda_k=0$. Therefore, $[\bm f, \bm\lambda]$ satisfies all the KKT conditions.
By Lemma \ref{lemma::kkt}, $\bm f$ is an optimal solution to (OP2) and thus it determines a point on the outer boundary of a reliable throughput region. 
%
%
\end{proof}


\begin{example}
Continue on Example~\ref{example::counterex}.
When $\bm f=[1/2,1/2,1/2]$, we have \begin{align*}
    \bm F(\bm E+\bm I) &= \begin{bmatrix}
        1/2&1/2&1/2\\
        1/2&1/2&0\\
        1/2&0&1/2
    \end{bmatrix},
\end{align*} of which the Perron–Frobenius eigenvalue $\rho=(1+\sqrt{2})/2\neq 1$. Hence, the corresponding throughput vector $[1/8,1/4,1/4]$ is not on the outer boundary of $\mc C$.
Let $\bm F' = \text{diag}(\bm f/\rho)$. Then $\bm F'(\bm E+\bm I)$ has the Perron–Frobenius eigenvalue one, indicating that the corresponding throughput vector $[2/(1+\sqrt{2})^3, \sqrt{2}/(1+\sqrt{2})^2, \sqrt{2}/(1+\sqrt{2})^2]$ is on the outer boundary of $\mc C$.
%
\end{example}

For multiple access communications considered in \cite{massey1985collision}, $\bm E+\bm I$ is the matrix of ones. The condition that the Perron-Frobenius eigenvalue of $\bm F(\bm E+\bm I)$ is one is equivalent to $\sum_{i=1}^M f_i = 1$, which is consistent with the conclusion in \cite{massey1985collision}.

\section{Concluding remarks}\label{section:conclusion}

We discuss communication systems with a massive number of users in this section. When all links employ the same duty factor $f$, the throughputs $T_i = f(1-f)^{|\mc I(l_i)|}$ for all $i=1,\ldots, M$. Then the sum throughput of the system is $\sum_{i=1}^M f(1-f)^{|\mc I(l_i)|}$.
Consider a special case that $|\mc I(l_i)|=N-1$ for all links. Then $f=1/N$ maximizes both individual and sum throughputs, and $\sum_i T_i = M(1-1/N)^{N-1}/N$. When the number of users $M$ tends to infinity, if $N$ is a constant not scaling with $M$, then the sum throughput also tends to infinity; if $N$ scales with $M$ linearly, e.g., $N=M/a$ for some constant $a\in\mathbb{Z}_{++}$, then the sum throughput limit becomes $a/e$.

\newpage
\bibliographystyle{IEEEtran}
\bibliography{isit_final_ver.bib}

\begin{thebibliography}{10}
\providecommand{\url}[1]{#1}
\csname url@samestyle\endcsname
\providecommand{\newblock}{\relax}
\providecommand{\bibinfo}[2]{#2}
\providecommand{\BIBentrySTDinterwordspacing}{\spaceskip=0pt\relax}
\providecommand{\BIBentryALTinterwordstretchfactor}{4}
\providecommand{\BIBentryALTinterwordspacing}{\spaceskip=\fontdimen2\font plus
\BIBentryALTinterwordstretchfactor\fontdimen3\font minus
  \fontdimen4\font\relax}
\providecommand{\BIBforeignlanguage}[2]{{%
\expandafter\ifx\csname l@#1\endcsname\relax
\typeout{** WARNING: IEEEtran.bst: No hyphenation pattern has been}%
\typeout{** loaded for the language `#1'. Using the pattern for}%
\typeout{** the default language instead.}%
\else
\language=\csname l@#1\endcsname
\fi
#2}}
\providecommand{\BIBdecl}{\relax}
\BIBdecl

\bibitem{ephremides1998information}
A.~Ephremides and B.~Hajek, ``Information theory and communication networks: An
  unconsummated union,'' \emph{IEEE Transactions on Information Theory},
  vol.~44, no.~6, pp. 2416--2434, 1998.

\bibitem{massey1985collision}
J.~Massey and P.~Mathys, ``The collision channel without feedback,'' \emph{IEEE
  Transactions on Information Theory}, vol.~31, no.~2, pp. 192--204, 1985.

\bibitem{shum2009shift}
K.~W. Shum, C.~S. Chen, C.~W. Sung, and W.~S. Wong, ``Shift-invariant protocol
  sequences for the collision channel without feedback,'' \emph{IEEE
  Transactions on Information Theory}, vol.~55, no.~7, pp. 3312--3322, 2009.

\bibitem{wong2007new}
W.~S. Wong, ``New protocol sequences for random-access channels without
  feedback,'' \emph{IEEE Transactions on Information Theory}, vol.~53, no.~6,
  pp. 2060--2071, 2007.

\bibitem{chen2018crt}
Y.~Chen, Y.-H. Lo, K.~W. Shum, W.~S. Wong, and Y.~Zhang, ``{CRT} sequences with
  applications to collision channels allowing successive interference
  cancellation,'' \emph{IEEE Transactions on Information Theory}, vol.~64,
  no.~4, pp. 2910--2923, 2018.

\bibitem{tinguely2005collision}
S.~Tinguely, M.~Rezaeian, and A.~J. Grant, ``The collision channel with
  recovery,'' \emph{IEEE Transactions on Information Theory}, vol.~51, no.~10,
  pp. 3631--3638, 2005.

\bibitem{bae2014outerbdmr}
Y.~H. Bae, B.~D. Choi, and A.~S. Alfa, ``Achieving maximum throughput in random
  access protocols with multipacket reception,'' \emph{IEEE Transactions on
  Mobile Computing}, vol.~13, no.~3, pp. 497--511, 2014.

\bibitem{zhang2016protocol}
Y.~Zhang, Y.-H. Lo, W.~S. Wong, and F.~Shu, ``Protocol sequences for the
  multiple-packet reception channel without feedback,'' \emph{IEEE Transactions
  on Communications}, vol.~64, no.~4, pp. 1687--1698, 2016.

\bibitem{zhang2020throughtputone}
Y.~Zhang, Y.~Chen, Y.-H. Lo, and W.~S. Wong, ``Achieving zero-packet-loss
  throughput 1 for a collision channel without feedback and with arbitrary time
  offsets,'' \emph{IEEE Transactions on Information Theory}, vol.~66, no.~4,
  pp. 2269--2279, 2020.

\bibitem{ma2021rate}
J.~Ma, Y.~Liu, and S.~Yang, ``Rate region of scheduling a wireless network with
  discrete propagation delays,'' in \emph{IEEE INFOCOM 2021-IEEE Conference on
  Computer Communications}.\hskip 1em plus 0.5em minus 0.4em\relax IEEE, 2021,
  pp. 1--10.

\bibitem{fan2021continuity}
Y.~Fan, Y.~Liu, and S.~Yang, ``Continuity of link scheduling rate region for
  wireless networks with propagation delays,'' in \emph{2022 IEEE International
  Symposium on Information Theory (ISIT)}, 2022, pp. 730--735.

\bibitem{abramson1973packet}
N.~Abramson, ``Packet switching with satellites,'' in \emph{Poceedings of the
  June 4-8, 1973, national computer conference and exposition}, 1973, pp.
  695--702.

\bibitem{axelsson_1994}
O.~Axelsson, \emph{Reducible and Irreducible Matrices and the Perron-Frobenius
  Theory for Nonnegative Matrices}.\hskip 1em plus 0.5em minus 0.4em\relax
  Cambridge University Press, 1994, p. 122–157.

\end{thebibliography}


\end{document}